\documentclass [12pt,reqno]{amsart}
\usepackage{amsmath}
\usepackage{amsthm}
\usepackage{dsfont}
\newtheorem{thm}{Theorem} 
\usepackage{amsfonts, amssymb}
\usepackage[utf8]{inputenc} %% mac coding
\usepackage[english]{babel}
\usepackage{mathtools}
\usepackage{ mathrsfs }
\usepackage{enumerate}
\usepackage{ amssymb }
\usepackage{ bbold }
\usepackage{cite}
\usepackage{hyphenat}

\newtheorem{theorem}{Theorem}

\newtheorem{Corollary}[theorem]{Corollary}

\newtheorem*{theorem*}{Theorem}

\newtheorem{proposition}[thm]{Proposition}

\theoremstyle{definition}

\theoremstyle{definition}

\theoremstyle{definition}

\newcommand{\SINR}{\operatorname{SINR}}
\newcommand{\rank}{\operatorname{rank}}

\numberwithin{equation}{section}

\def\R{{\mathbb R}}
\def\ww{{\omega}}

\def\C{{\mathbb C}}

\def\<{\langle}
\def\>{\rangle}
\def \ee{{\epsilon}}
\def \dd{{\delta}}
\def \ss{{\sigma}}

\def \aa {{\alpha}}
\def \bb {{\beta}}
\def \ll {{\lambda}}

\begin{document}

\title{Upper bounds for the number of isolated critical points via Thom-Milnor theorem}
%\title{Markov type $\bold 2$ with constant $\bold 1$}

\begin{abstract}
We apply the Thom-Milnor theorem to obtain the upper bounds on the amount of 
isolated (1) critical points of a potential generated by several fixed point 
charges(Maxwell's problem on point charges), (2) critical points of SINR, 
(3) critical points of a potential generated by several fixed  
Newtonian point masses augmented with a quadratic term, 
(4) central configurations in the $n$-body problem. 
In particular, we get an exponential 
bound for Maxwell's problem and the polynomial 
bound for the case of an "even dimensional" potential in Maxwell's problem.
%Maxwell's problem asks for an upper bound on the number of isolated
%critical points of an electrostatic potential generated by $n$ fixed point 
%charges in $\mathbb{R}^3$. We show that the critical points of such a potential are in one-to-one correspondence 
%with solutions of a certain polynomial system. 
%As a corollary we obtain an exponential upper bound for the number of isolated critical points and 
%a polynomial upper bound for the case of an "even dimensional" electrostatic potential.

\end{abstract}
\keywords{Newtonian potential, point charges, points of equilibrium, SINR, central configurations}
\subjclass[2010]{31B05}

\author{Vladimir Zolotov}
%\address[Vladimir Zolotov]{No need in a job atm.}
\email[Vladimir Zolotov]{lemiranoitz@gmail.com}

\maketitle
% ----------------------------------------------------------------

\section{Introduction}

In the present paper, we use the Thom-Milnor theorem to 
give upper bounds for the number of isolated  
\begin{enumerate}
\item{critical points of a potential of generated by several fixed point 
charges (Maxwell's problem on point charges),\label{PL-Maxwell}} 
\item{critical points of $\SINR$,\label{PL-SINR}} 
\item{critical points of a potential generated by several fixed  
Newtonian point masses augmented with a quadratic term,\label{PL-Newton}}
\item{central configurations in the $n$-body problem.\label{PL-Central}}
\end{enumerate}

Surprisingly the direct application of the Thom-Milnor theorem 
allows to obtain tighter or more general bounds 
than specialized methods. In particular, we get an exponential 
bound for (\ref{PL-Maxwell}) 
% and (\ref{PL-Newton})
which were previously 
only known for the $2$-dimensional case, see \cite{killian2009remark}. Additionally for the case of 
"even dimensional" potential in (1) we get a polynomial bound.

\subsection{Thom-Milnor theorem}
Our main tool is the following theorem by R. Thom \cite{thom2015homologie} 
and J. Milnor \cite[Theorem 2]{milnor1964betti}. 
\begin{proposition} \label{ThomMilner}
Let $m,k,p = 1,2,3,\dots$ and $f_1,\dots,f_p$ be a real
polynomials in $m$ variables. Let $V$ be the 
zero set of the system 
\begin{equation}f_1(x_1,\dots,x_m) = \dots = f_p(x_1,\dots,x_m) = 0. 
\label{TM-sys}\end{equation} 
Suppose that each $f_i$ has degree $\le k,$ then 
the sum of the Betti numbers of $V$ is $\le k(2k - 1)^{m-1}$. 
\end{proposition} 
In the above formulation by $q$th Betti numbers of $V$ 
we mean the rank of Čech cohomology group $H^q(V)$,
using coefficients in some fixed field $F$. All we need 
from Betty numbers are two following properties:
\begin{itemize}
\item{
 they all are non-negative since they are ranks 
of some groups,
}
\item{
$0$th Betti number $b_0$ is the number of the connected components of $V$.
}
\end{itemize}
The above two properties directly imply the following corollary.
\begin{Corollary} \label{ThomMilnerCor}
Under assumptions of Proposition \ref{ThomMilner},
the number of connected components of $V$
does not exceed 
$$k(2k - 1)^{m-1},$$
and in particular the number of isolated zeroes 
of the system (\ref{TM-sys}) does not exceed 
$$k(2k - 1)^{m-1}.$$
\end{Corollary} 

%\subsection{Bezout's theorem}
%We also use the following version of Bezout's theorem, see \cite[Example 8.4.6]{fulton2013intersection}.
%\begin{proposition} \label{BezoutThm}
%Let $f_1,\dots,f_n$ be homogeneous polynomials in an $n$-dimensional complex projective space. 
%Then the number of connected components of the set of their common zeroes doesn't exceed \
%the product of their degrees.
%\end{proposition}
%More precisely we will use the following corollary. 
%\begin{Corollary} \label{BezoutCor}
%Let $m,p = 1,2,3,\dots$ and $f_1,\dots,f_m$ be a real
%polynomials in $m$ variables.
%Then, the number of non-degenerate real solutions 
%%Let $V$ be the zero set 
%of the system 
%\begin{equation}f_1(x_1,\dots,x_m) = \dots = f_m(x_1,\dots,x_m) = 0  
%\label{Bez-sys}\end{equation} 
% does not exceed 
% $$\prod_{i = 1}^{i = m}\deg(f_i),$$
% where $\deg(f_i)$ denotes the degree of $f_i$.
%\end{Corollary}
%\begin{proof}
%Note that for each non-degenerate real solution of the system (\ref{Bez-sys})
%the connected component of zeroes of the same system in $\C^m$ consists 
%of a single point, namely this solution. If our system is not homogeneous 
%we can homogenize it and the rest follows from Proposition \ref{BezoutThm}.
%\end{proof}

\subsection{The structure of the paper}
Each of Sections \ref{SecMaxwell}-\ref{SecCentral} is devoted to
a single  
problem from (\ref{PL-Maxwell}) - (\ref{PL-Central}). 
At the beginning of each section, We give a proper introduction to the problem and  
provide some historical remarks. Then we derive the bounds 
by applying Corollary \ref{ThomMilnerCor}. Finally, at the 
end of each section, we discuss the 
(non-)existence of non-isolated solutions. For problems 
(\ref{PL-Maxwell}), (\ref{PL-SINR}) and (\ref{PL-Newton}) for almost all sets of parameters 
there are no degenerate critical points and in particular 
there are no non-isolated critical points, see Proposition \ref{morseAlmostAll}, Proposition \ref{SINRAlmostAll}
and the discussion in Subsection \ref{NewtonNonIso}. 
There is no such result known for (\ref{PL-Central}). And if one could establish it
that would be a big result, since the question of whether non-isolated central configurations 
exist is a long-standing open problem known as Smale's $6$th problem, see \cite{smale1998mathematical}.

\subsection{Previously known connections between problems (\ref{PL-Maxwell}) - (\ref{PL-Central})}
The author is not the first person who noticed that problems (\ref{PL-Maxwell}) - (\ref{PL-Central})
are connected to each other. The problem (\ref{PL-Newton}) originates 
in the study of central configurations, see \cite{arustamyan2021number}, 
and thus is connected to (\ref{PL-Central}). The similarity between (\ref{PL-SINR})
and (\ref{PL-Maxwell}) is noted in \cite{kantor2011topology}.

The previously known bound for (\ref{PL-Maxwell}) by
A. Gabrielov, D. Novikov, and B. Shapiro \cite{Mystery} relies on the
Khovanskii's theory of fewnomials \cite{fewnomials}.
This paper inspired the usage of the theory of fewnomials in other 
works. In particular A. Albouy and Ya. Fu \cite{albouy2007euler}
apply it in the context of counting central configurations.
%Similar to Corollary \ref{ThomMilnerCor} Khovanskii's theory 
%allows one to obtain upper bounds on the number of certain roots of a polynomial 
%system. 
%The advantage of Corollary \ref{ThomMilnerCor} is that 
%the power in the exponent of the bound (which is the number of 
%independent variables minus $1$) is relatively small in comparison 
%to Khovanskii's which is the squared number of unique monomials of the system.    

%[FIX!!!]
%Another connection to Astronomy is via central configurations in the 
%Newtonian n-body problem, see \cite{albouy2008symmetry, albouy2007euler}. 

\section{Maxwell's problem} \label{SecMaxwell}
\subsection{The statement}
Fix $d \in \{1,2,3,\dots\}$.
Let $\vert \vert \cdot \vert \vert$ denotes the the standard Euclidean
norm. 
Fix $m \in \{0,1,2,\dots\}$. 
Suppose that we have points $x_1,\dots,x_n \in \R^d$.
And numbers $q_1,\dots,q_n \in \R \setminus \{0\}$ 
which symbolize point charges located in those points. 
%Those charges create an electrostatic field whose 
%potential in a given point $p \in \R^{d}$ is given by 
Consider a function 
$V = V_m^{(x_1,q_1),\dots,(x_n,q_n)}:\R^{d} \rightarrow \R$ given by
$$V(p) = \sum_{i = 1}^{n}\frac{q_i}{\vert \vert p - x_i \vert \vert^m}, \text{for $m \neq 0$,}$$
$$V(p) = \sum_{i = 1}^{n} {q_i}\log{\vert \vert p - x_i \vert \vert}, \text{for $m = 0$.}$$
 Maxwell's problem asks for an upper bound on the number of critical points of $V$. 

If $m = d - 2$ then $V$ is the potential of electrostatic field 
created by point charges $q_1,\dots,q_n$ (maybe up to a constant)
as it is considered in mathematical physics. If in addition 
$d = 3$ then $V$ is the "real life" potential of 
electrostatic field from physics. 

J.C. Maxwell \cite{maxwell1873treatise} argued that in the case $d = 3, m = 1$ 
the upper bound
on the number of isolated critical points of $V$ 
is $(n - 1)^2$, but 
%argument is unclear and seems to be incomplete 
his proof contains an unproven claim, 
and thus is considered to be incomplete   
(see \cite[Section 4]{Mystery}).
%Thus, the above upper bound is referenced as Maxwell's problem.

\subsection{Rough summary of our results} \label{MaxwellSummary}
%We show that 
If $m$ is even then the critical 
points of $V$ coincide with solutions 
of a system of polynomial equations with $d$ independent 
variables. Combining this with the
Thom-Milnor estimate on the number of connected 
components of a set of solutions of a polynomial system 
we get that the number of isolated critical points 
of $V$ is bounded from above by  
$$(1 + (n-1)(m+2))(1 + 2(n-1)(m+2))^{d-1},$$
%and the number of non-degenerate critical points of $V$
%is less or equal to 
%$$(1 + (n-1)(m+2))^d,$$
 see Theorem \ref{evenPoly}(\ref{evenPoly-p2}).

For the general case 
%we obtain that 
critical points of $V$
correspond to zeroes of a polynomial system with $n + d$ independent 
variables. Once again by the use of the Thom-Milnor estimate we get 
that the number of isolated critical points 
of $V$ can not exceed 
$$(m + 4)(2m + 7)^{d + n},$$
%and the number of non-degenerate critical points of $V$
%is less or equal to 
%$$4^n(m + 3)^d,$$
see Theorem \ref{evenPoly}(\ref{evenPoly-p4}).

%Which improves the bound from \cite{Mystery}. Although 
%is still very far from what we were promised by J.C. Maxwell.
%To be clear the author is very incompetent in estimating 
%solutions of polynomial systems and it is possible that 
%somebody less incompetent could get a much better bound.

%The key ingredient for the polynomial representation 
%is the fresh distance formula in 
%barycentric coordinates \cite{MO433272, zolotov2022scalar}.

%for homogeneous polynomials

%\section{Fixed objects}
%For the rest of the paper we fix points $x_1,\dots,x_n $ 
%which are in general position in 
%$(n-1)$-dimensional Euclidean space $(\R^{n-1}, \<\cdot,\cdot\>)$.
% We  denote by $D$ the distance matrix for $x_1,\dots,x_n$, 
%in other words $D_{ij} = \vert \vert x_i - x_j \vert \vert^2$.
%
 
%In the following we are going to consider barycentric coordinates
%with respect to points $x_1,\dots,x_n $.

\subsection{Previous results}
\subsubsection{The only previously known result for $d = 3$ and $m = 1$ and arbitrary $n$.}
%Before the current paper the only result which deals arbitrary $$
The work \cite{Mystery} by A. Gabrielov, D. Novikov, and B. Shapiro
is the only one which deals with arbitrary $d,m$ and $n$ and even 
the only one which addresses the case $d = 3$ and $m = 1$ for arbitrary $n$.
%There is one caveat which is that \cite{Mystery} only deals with  
%configurations of charges such that all critical points of 
%$V$ are 
%non-degenerate.
%The only known result which works for arbitrary $d,m$ and 
%A configuration of charges is called non-degenerate if all the 
%critical points of $V$ are non-degenerate. 
%In \cite{Mystery} 
%A. Gabrielov, D. Novikov and B. Shapiro represent the set 
%of critical points of $V$ as a set of solutions of 
%a system of quasipolynomial equations. 
Authors of \cite{Mystery} represent the set 
of critical points of $V$ as a set of solutions of 
a system of quasi-polynomial equations.  
%Using the theory
%of fewnomials developed by A. G. Khovanskii \cite{fewnomials} they 
%get that the total number
From that by application of Khovanskii's theory of fewnomials \cite{fewnomials}
they deduce that
for any $m \in \{0,1,2,\dots\}$ and 
any $d \in \{1,2,\dots\}$ if all the critical points of $V$ are non-degenerate then 
their total number
does not exceed $4^{n^2}(3n)^{2n}$.
%under the assumption of non-degeneracy of configuration.
%Note that this bound does not depend on $d$. 

In general, there exist configuration where critical points are not isolated:
consider a square with point charges $1, -1, 1, -1$ in its 
vertices. Then every point on the line through the center of the square 
and orthogonal to the plane of the square will be critical. It is 
unknown if non-isolated critical points could exist if all 
the charges have the same sign, see \cite[Conjecture 2]{shapiro2015problems}.  
\label{degen}
\subsubsection{Other results which work for arbitrary $n$.}\label{Killian}
In \cite{killian2009remark} K. Killian considers the case $d = 2, m = 1$ and 
shows that if all critical points of $V$ are isolated then their total number 
does not exceed $2^{2n-2}(3n-2)^2$. 

For the case $d = 2, m = 1$ it is unknown if non-isolated critical points are possible.
T. Erd{\'e}lyi,  J. Rosenblatt  and R. Rosenblatt \cite{erdelyi2021zero} 
show that there are no isolated critical points in the case if
all point charges are on the same line.

The study a of a case $d = 2, m = 0$ goes back to K. F. Gauss,
see \cite[Chapter I.3]{marden1949geometry}. 
In this case using the identification  $\R^2 \cong \C$ we can write 
$$(\nabla V(z))^* = c \sum_{i = 1}^{n} \frac{q_i}{( z - x_i )},$$
where $*$ denotes the complex conjugation and $c \neq 0$ is an absolute constant. 
Thus the zeroes of $\nabla V$ coincide with zeroes of a complex polynomial of 
a degree at most $(n - 1)$ and $V$ has at most $(n - 1)$ critical points.

\subsubsection{Modeling of $m = m'$ inside $m = m' + 1$.}
Another phenomena mentioned in \cite[Chapter I.3]{marden1949geometry} is that
%for any $m_0 \in \{0,1,2,\dots\}$ and any $d_0 \in \{1,2,\dots\}$
%one can model $V_m$ in $\R^{d_0}$ with $V_{m+1}$ in $\R^{d+1}$.
one can generate $$V = V_m^{(x_1,q_1),\dots,(x_n,q_n)}:\R^{d} \rightarrow \R$$ 
using $V_{m+1}$ by adding additional dimension and substituting point charges by charged lines. 
More precisely consider 
$$\bar V = \bar V_{m + 1}^{(l_1,q_1),\dots,(l_n,q_n)}:\R^{d} \times \{0\} \rightarrow \R,$$
where $l_i$ is a line $(x_i,*)$ and $\bar V$ is given by 
$$\bar V(p) = \sum_{i = 1}^{n} \int_{x \in l_i}\frac{q_i dx}{\vert \vert p - x \vert \vert^{m+1}}.$$
Then $V = c(m) \bar V$, where $c(m) > 0$ is a constant depending only on $m$.

Theorem \ref{evenPoly} gives a polynomial upper bound for 
the number of isolated critical points for the case of even $m$. 
Thus, if one could model $m = m'$ inside $m = m' + 1$ using say $100$ point charges 
instead of each charged line then he would likely be able to get a polynomial upper bound for
the case of odd $m$ too. But the author fails to figure 
out how to do a reduction of this nature.

%this would allow us to get a polynomial upper bound 
%for the case of odd $m$ from the polynomial upper bound for the case of even $m$
%given by Theorem \ref{evenPoly}.
 
%\subsection{Lower bounds on the number of critical points.} 
\subsubsection{Results for a specific $n$.} 
T.-L. Lee and Y.-L. Tsai \cite{lee2022nine} give an example with $9$ equilibrium 
points for $d = 2, m = 1$, and $n = 4$ which is the claimed upper bound of the Maxwell conjecture.

\subsubsection{Lower bounds for the number of isolated critical points.}
Lower bounds on the number of critical points of $V$ are given
by M. Morse and S. Cains \cite[Theorem 32.1]{morse2014critical} and T. Kiang \cite[Theorem 6]{kiang1932critical}. 

%In wireless communication networks   

\subsubsection{Critical points of polynomials} Khavinson et al. \cite{khavinson2011borcea} formulated conjectures on location 
of critical points of polynomials related to Maxwell's problem.

%In the above we describe a 
%rather specialized version of the result from \cite{Mystery}.
%One other situation when the same conclusion holds is 
%when $x_1,\dots,x_n$ are points in  $\R^{n-1}$. 

\subsection{Our results for Maxwell's problem}
The following theorem formalizes the informal summary given in Subsection \ref{MaxwellSummary}. 

\begin{theorem}\label{evenPoly}
Let $x_1 = (x_{11},\dots,x_{1d}),\dots,x_{n} = (x_{n1},\dots,x_{nd})$ 
be distinct points in $\R^d$,
$p  = (p_1,\dots,p_d) \in \R^d \setminus \{x_1,\dots, x_n\}$,
$m \in \{0,1,2,3,\dots\}$, 
$q_1,\dots,q_n \in \R \setminus \{0\}$
and $V = V_m^{(x_1,q_1),\dots,(x_n,q_n)}$. Then,
\begin{enumerate}
\begin{item}\label{evenPoly-p1}
Point $p$ is a critical point of $V$ iff 
\begin{equation}
\sum_{i=1}^{n}
\Big(
%q_i(p_m - x_{im})
q_i(p - x_{i})
\prod_{\substack {  1 \le j \le n \\ j \neq i} }
(\sum_{1\le k \le d}(p_k - x_{jk})^2)^{\frac{m+2}{2}}
\Big) = 0.
%\text{, for every $1\le m \le d$.}
\label{EEE1}\end{equation}
\end{item}
\begin{item}
If $m$ is even then $V$ has at most 
$$(1 + (n-1)(m+2))(1 + 2(n-1)(m+2))^{d-1}$$
isolated critical points.
%and out of them at most 
%$$(1 + (n-1)(m+2))^d$$
%are non-degenerate critical points.
\label{evenPoly-p2}
\end{item}
\begin{item}
Point $p$ is a critical point of $V$ iff there exist $\ss_1,\dots,\ss_n > 0$
satisfying
%such that $p_1,\dots,p_d,\ss_1,\dots,\ss_n$ satisfy 
$$
\ss_j^2 \sum_{1\le k \le d}(p_k - x_{jk})^2 = 1
\text{, for every $1\le j \le n$ and}
$$
\begin{equation}
\sum_{i=1}^{n}
\Big(
%q_i(p_m - x_{im})
q_i(p - x_{i})
\ss_i^{m+2}
%\prod_{\substack {  1 \le j \le n \\ j \neq i} }
% \ss_j^{m+2}
\Big) = 0.
%\text{, for every $1\le m \le d$.}
\label{EEE2221}\end{equation}\label{evenPoly-p3}
\end{item}
\begin{item}
$V$ has at most 
$$(m + 4)(2m + 7)^{d + n}$$
 isolated critical points. 
% and out of them at most 
% $$4^n(m + 3)^d$$
% are non-degenerate critical points.
\label{evenPoly-p4}
\end{item}
\end{enumerate}
\end{theorem}

%%%%%%%%%%%%%%%%%%%%%%%%%%%%%%%%%%
%%%%%%%%%%%%%%%%%%%%%%%%%%%%%%%%%
%%%%%%%%%%%%%%%%%%%%%%%%%%%%%%kkkk iiii

\begin{proof}[Proof of Theorem \ref{evenPoly}(\ref{evenPoly-p1})]
We remind that  
$$V(p) = \sum_{i = 1}^{n}\frac{q_i}{\vert \vert p - x_i \vert \vert^m}, \text{for $m \neq 0$,}$$
$$V(p) = \sum_{i = 1}^{n} {q_i}\log{\vert \vert p - x_i \vert \vert}, \text{for $m = 0$.}$$
We differentiate $V$:
$$V'_{p_k}  =  c(m)\sum_{i = 1}^{n}
\frac{q_i(p_k - x_{ik})}
{\vert \vert p - x_i \vert \vert^{m + 2}},$$where $c(m) \neq 0$
is a constant depending only on $m$,  
$$\nabla V =  c(m)\sum_{i = 1}^{n}
\frac{q_i(p - x_{i})}
{\vert \vert p - x_i \vert \vert^{m+2}},$$
which is (maybe up to a constant) the force given by Coulomb's law.
Thus, $\nabla V(p) = 0$ is equivalent to  
$$
\sum_{i=1}^{n}
\Big(
%q_i(p_m - x_{im})
q_i(p - x_{i})
\prod_{\substack {  1 \le j \le n \\ j \neq i} }
\vert \vert p - x_j \vert \vert^{ m+2 }
\Big) = 0.$$
Which is equivalent to (\ref{EEE1}). 
\end{proof}

\begin{proof}[Proof of Theorem \ref{evenPoly}(\ref{evenPoly-p2})]
When $m$ is even (\ref{EEE1}) is a polynomial system in $d$
variables: $p_1,\dots,p_d$. Each polynomial of the system 
has degree $\le (1 + (n-1)(m+2))$. Thus, by Thom-Milnor 
theorem (see Corollary \ref{ThomMilnerCor}) we have the number of 
isolated critical points of $V$ does not exceed
$$(1 + (n-1)(m+2))(1 + 2(n-1)(m+2))^{d-1}.$$
%By Bezout's theorem see Corollary \ref{BezoutCor} the number
%of non-degenerate critical points is less or equal to 
%$$(1 + (n-1)(m+2))^d$$.
\end{proof}

\begin{proof}[Proof of Theorem \ref{evenPoly}(\ref{evenPoly-p3})]
Theorem \ref{evenPoly}(\ref{evenPoly-p3}) is 
just a reformulation of 
Theorem \ref{evenPoly}(\ref{evenPoly-p1}).
\end{proof}
 
\begin{proof}[Proof of Theorem \ref{evenPoly}(\ref{evenPoly-p4})]
Denote 
 $$N =  (m + 4)(2m + 7)^{d + n}.$$
We will argue by contradiction. 
Suppose that $V$ has at least $N + 1$ isolated critical points.
Then by Theorem \ref{evenPoly}(\ref{evenPoly-p3})
the zero set of the system (\ref{EEE2221}) (considered as a polynomial 
system in $d + n$ variables $p_1,\dots,p_k,\ss_1,\dots,\ss_n$)
has at least $N + 1$ connected components. 

The degree of every polynomial in this system is 
$\le \max\{4, m+3\} \le m + 4$.
Thus,
by Thom-Milnor theorem (see Corollary \ref{ThomMilnerCor}) we 
have that the total number of the connected components of the zero set 
does not exceed 
$$(m + 4)(2m + 7)^{d + n} = N,$$
so we have a contradiction.   
\end{proof}

\subsection{Existence of non-isolated critical points}
As we already mentioned in subsubsections \ref{degen} and \ref{Killian} 
non-isolated critical points do exist but it is unknown if they exist 
in dimension $2$ or for potentials generated by charges 
of the same sign.

It is also known that in almost all configurations of 
charges there are no non-isolated critical points. 
More precisely M. Morse and S. Cains 
\cite[Theorem 6.2]{morse2014critical} give the following theorem. 

\begin{proposition} \label{morseAlmostAll}
Let $d \in \{1,2,3,\dots\}$ and $m = \{0,1,2,3,\dots\}$. 
Let $x_1 = (x_{11},\dots,x_{1d}),\dots,x_{n - 1} = (x_{(n - 1)1},\dots,x_{(n - 1)d})$ 
be distinct points in $\R^d$, and
$q_1,\dots,q_n \in \R \setminus \{0\}$.
Then, for almost all $x_{n} = (x_{n1},\dots,x_{nd}) \in \R^d \setminus \{x_1,\dots,x_{n-1}\}$ the potential
$V = V_m^{(x_1,q_1),\dots,(x_n,q_n)}$ has no degenerate critical points.  
\end{proposition}

The above proposition is a corollary of the following 
theorem \cite[Theorem 6.3]{morse2014critical}. 

\begin{proposition} \label{morseGen}
Let $d,\dd \in \{1,2,3,\dots\}$. Let $\R^{d+\dd}$ be a Euclidean space 
with Cartesian coordinates $x_1,\dots,x_d,a_1,\dots,a_\dd$. 
Let $W$ be an open non-empty subset of $\R^{d+\dd}$ and 
$U:W\rightarrow \R$ be a $C^2$-mapping such that for every point $p \in W$ 
satisfying 
$$U'_{x_i}(p) = 0\text{, for every $1 \le i \le d$},$$
we have  
$$\rank 
\begin{bmatrix}
    U''_{x_1x_1}       & U''_{x_1x_2} & \dots & U''_{x_1x_d} & U''_{x_1a_1} & U''_{x_1a_2} & \dots & U''_{x_1a_\dd} \\
    U''_{x_2x_1}       & U''_{x_2x_2} & \dots & U''_{x_2x_d} & U''_{x_2a_1} & U''_{x_2a_2} & \dots & U''_{x_2a_\dd} \\
    \hdotsfor{8} \\
    U''_{x_dx_1}       & U''_{x_dx_2} & \dots & U''_{x_dx_d} & U''_{x_da_1} & U''_{x_da_2} & \dots & U''_{x_da_\dd}
\end{bmatrix}
 = d.$$
Then for almost all $a \in \{\aa \in \R^\dd \vert \exists x \in \R^d: (x,\aa) \in W\}$ the map 
$$u^a: \{x \in \R^d \vert (x, a) \in W\} \rightarrow \R$$
given by 
$$u^a(x) = U(x,a)$$
does not have degenerate critical points (i.e., if the gradient of $u^a$
is $0$ at some point then Hessian of $u^a$ has the maximal rank at this point).
\end{proposition}
We will not give the proof of Proposition \ref{morseGen} since 
it can be found in \cite{morse2014critical}. But we will 
give the proof of Proposition \ref{morseAlmostAll} because 
\cite{morse2014critical} only states Proposition \ref{morseAlmostAll}
for $d = 3, m = 1$ and only sketches the proof. 
\begin{proof}[Proof of Proposition \ref{morseAlmostAll}]
We take $W = \{(p,a)\vert p,a \in \R^d, p \neq a, p \neq x_1,\dots, p \neq x_{n-1} \}$
and we define $$U:W \rightarrow \R$$
by $$U(p,a) = V_m^{(x_1,q_1),\dots,(x_{n-1},q_{n-1}), (a,q_n)}(p).$$
Consider a matrix $d \times d$ matrix $M$ given by  
$$M = 
\begin{bmatrix}
    U''_{p_1a_1} & U''_{p_1a_2} & \dots & U''_{p_1a_d} \\
    U''_{p_2a_1} & U''_{p_2a_2} & \dots & U''_{p_2a_d} \\
    \hdotsfor{4} \\
     U''_{p_da_1} & U''_{p_da_2} & \dots & U''_{p_da_d}
\end{bmatrix}
$$
and we are going to show that $\rank(M) = d$.
For $1 \le k \le d$ have 
$$U'_{a_k}(p) =  c(m)
\frac{q_n(p_k - a_{k})}
{\vert \vert p - a \vert \vert^{m+2}},$$
where $c(m) = m$, for $m \neq 0$ and $c(0) = -1$.
And thus for $1 \le j \neq k \le d$ we have 
$$U''_{a_kp_j}(p) =  -c(m)(m+2)
\frac{q_n(p_k - a_{k})(p_j - a_{j})}
{\vert \vert p - a \vert \vert^{m+4}},$$
and 
$$U''_{a_kp_k}(p) =  -c(m)(m+2)
\frac{q_n(p_k - a_{k})^2}
{\vert \vert p - a \vert \vert^{m+4}}  +c(m)
\frac{q_n}
{\vert \vert p - a \vert \vert^{m+2}}.$$
Thus 
\begin{equation} M = (c(m)
\frac{q_n}
{\vert \vert p - a \vert \vert^{m+2}})(I - (m+2)vv^T), \label{ImW}\end{equation}
where $$v = (p - a) / \vert \vert p - a \vert \vert.$$
Suppose that $\rank(M) < d$ then there exists $w \neq 0$ such that 
$$Mw = 0,$$
from (\ref{ImW}) we see that $w = \ll v$ for some $\ll \neq 0$.
But then 
$$Mw = M\ll v = (c(m)\frac{q_n}
{\vert \vert p - a \vert \vert^{m+2}})(\ll v - \ll v (m+2)) \neq 0.$$
So we have a contradiction and thus $\rank(M) = d$.
Now we can apply Proposition \ref{morseGen} and get that 
for almost all $a$ the potential $V_m^{(x_1,q_1),\dots,(a,q_n)}$
does not have degenerate critical points.
%and hence all it's critical points are isolated.
\end{proof}

\section{Critical points of $\SINR$}
\subsection{The statement}
In wireless communications the signal\hyp{}to\hyp{}interference\hyp{}plus\hyp{}noise ratio ($\SINR$)
is used as a way to measure the quality of wireless connection (see \cite{avin2009sinr, kantor2011topology}).  
Given $d \in \{2,3\}$, points $x_1,\dots,x_n \in \R^d$, 
$\psi_1,\dots,\psi_n > 0$, $\alpha > 0$ and $N \ge 0$ 
%representing the noise  
%real numbers representing transmitting powers of the devices at points $x_1,\dots,x_n$ 
a function $\SINR(x_i, \cdot):\R^d \setminus \{x_1,\dots,x_n\} \rightarrow \R$ is defined by 
\begin{equation} \SINR(x_i, p) = \frac{\psi_i \vert  \vert x_i - p \vert \vert^{-\alpha} }
{\sum_{j\neq i}(\psi_j \vert  \vert x_j - p \vert \vert^{-\alpha}) + N}. \label{SINRdef}\end{equation}   
In this model a receiver at a point $p$ successfully receives  
a message from sender $x_i$, if and only if $\SINR(x_i, p) \ge \beta,$
where $\beta$ is a constant $\ge 1$. 
Numbers $\psi_1,\dots,\psi_n$ represent transmitting powers of concurrently transmitting stations
at points $x_1,\dots,x_n$. $N$ represents the environmental noise. 
The pass-loss parameter $\alpha$ is typically taken from the interval $[2,4]$, with $\aa = 2$ being the 
most common. The reception threshold $\bb$ is commonly taken to be $\bb \approx 6$.

Interest in the critical points of $\SINR$ is motivated by the point location problem, see \cite{kantor2011topology}.
Here is the description of the problem from \cite{kantor2011topology}:
"Given a query point $p$, it is required to identify which of the $n$ transmitting stations is heard at $p$,
if any, under interference from all other $n-1$ transmitting stations and background noise $N$.
Obviously, one can directly compute $\SINR(x_i, p)$ for every $i \in \{1,\dots, n\}$ in time $\Theta(n)$
and answer the above question accordingly. Yet, this computation may be too expensive,
if the query is asked for many different points $p$."

\subsection{Rough summary of our results for $\SINR$} \label{SINRSummary}
%We show that 
For the case when $\aa$ is an even integer the critical 
points of $\SINR(x_i, \cdot)$ coincide with solutions 
of a system of polynomial equations with $d$ independent 
variables. From the  
Thom-Milnor Theorem (see Corollary \ref{ThomMilnerCor}) we get that the number of isolated critical points
is bounded from above by 
$$(\aa (2n - 1)  - 1)(2\aa(2n - 1)  - 3)^{d-1},$$
 see Theorem \ref{SINRThm}(\ref{SINRThm-bound}).
 
 %For almost all sets of locations of transmitters 
 %$\SINR(x_i, \cdot)$ does not have non-isolated critical points.

\subsection{Our results for $\SINR$}
Now we give a more detailed version of the above statement. 

\begin{theorem}\label{SINRThm}
Let $d \in \{2,3\}$, let $x_1,\dots,x_{n}$ be points in $\R^d$,
$p  = (p_1,\dots,p_d) \in \R^d \setminus \{x_1,\dots, x_n\}$,
$\aa \in \{2,4,6,8\dots\}$, 
$\psi_1,\dots,\psi_n > 0$, $N \ge 0$,
$i \in \{1,\dots, n\}$
and $\SINR(x_i, \cdot)$ be a function defined as in (\ref{SINRdef}). Then,
\begin{enumerate}
\begin{item}\label{SINRThm-1}
Point $p$ is a critical point of $\SINR(x_i, \cdot)$ iff for every $m = {1,\dots,d}$ we have 
 
\begin{equation} f'_m(p) g(p) - f(p) g'_m(p) = 0,\label{SINR-EEE} \end{equation}
where $f,g$ are real polynomials in variables $(\bar p_1,\dots,\bar p_d) := \bar p$ given by 
$$f = \psi_i \prod_{j \neq i}\vert \vert x_i - \bar p \vert \vert^\aa,$$ 
$$g = \sum_{j \neq i}(\psi_j \prod_{k \neq j}\vert \vert x_k - \bar p \vert \vert^\aa)
+ N \prod_{1 \le k \le n}\vert \vert x_k - \bar p \vert \vert^\aa.$$ 
\end{item}
\begin{item}
$\SINR(x_i, \cdot)$ has at most 
$$(\aa (2n - 1)  - 1)(2\aa(2n - 1)  - 3)^{d-1}$$
 isolated critical points.
\label{SINRThm-bound}
\end{item}
\end{enumerate}
\end{theorem} 
\begin{proof}[Proof of Theorem \ref{SINRThm}(\ref{SINRThm-1})]
Note that 
$$\SINR(x_i, \bar p) = \frac{f(\bar p)}{g(\bar p)}.$$
Thus, 
$$\SINR'_m(x_i, p) = \frac{f'_m(p) g(p) - f(p) g'_m(p)}{g^2(p)}$$
and the claim follows. (Note that since $\psi_1,\dots,\psi_n > 0$ the denominator 
is always positive and thus can not create any problems.)
\end{proof}
\begin{proof}[Proof of Theorem \ref{SINRThm}(\ref{SINRThm-bound})]
Since $\aa$ is even natural number, (\ref{SINR-EEE}) is a polynomial system in $d$
variables: $p_1,\dots,p_d$. Each polynomial of the system 
has degree $\le (\aa (2n - 1)  - 1)$. Thus, by Thom-Milnor 
theorem (see Corollary \ref{ThomMilnerCor}) we have the number of 
isolated critical points of $V$ does not exceed
$$(\aa (2n - 1)  - 1)(2\aa(2n - 1)  - 3)^{d-1}.$$ 
\end{proof}

\subsection{Existence of non-isolated critical points}
It is unclear whether non-isolated critical points of $\SINR(x_i, \cdot)$
 could exist in the general case. 
But for almost all of the selections of locations of transmitters
there are no non-isolated critical points.
The following proposition has the same nature as Proposition \ref{morseAlmostAll} 
but the proof is a bit more messy. 

\begin{proposition} \label{SINRAlmostAll}
Let $d \in \{2,3\}$, $n \in \{d+2, d+3, d+4,\dots\}$, $i \in \{1,\dots, n\}$,
let $x_i$ be a point in $\R^d$, 
%let $x_1,\dots,x_{i-1},x_{i+1},\dots,x_{n}$ be points in $\R^d$, 
$\aa \in \{2,4,6,8,\dots\}$, 
$\psi_1,\dots,\psi_n > 0$, $N \ge 0$.
Then, for almost all 
$(x_1 = (x_{11},\dots,x_{1d}),\dots,x_{i-1} = (x_{(i-1)1},\dots,x_{(i-1)d}),
x_{i+1} = (x_{(i+1)1},\dots,x_{(i+1)d}),
\dots,x_{n}= (x_{n1},\dots,x_{nd})) \in \R^{d(n-1)}$
the function  %= (x_{n1},\dots,x_{nd}) 
$\SINR(x_i, \cdot)$ has no degenerate critical points.  
\end{proposition}

\begin{proof}[Proof of Proposition \ref{SINRAlmostAll}]
Instead of working of working with $\SINR(x_i, \cdot)$ we prefer to work 
with $$\frac{1}{\SINR(x_i, \cdot)}  = 
\frac{\sum_{j\neq i}(\psi_j \vert  \vert x_j - p \vert \vert^{-\alpha}) + N}
{\psi_i \vert  \vert x_i - p \vert \vert^{-\alpha} }.$$ By computing the gradient and the 
Hessian one can see that for any open $\Omega \subset \R^d$ and any
 $C^2$-function $f:\Omega \rightarrow (0,\infty)$
\begin{enumerate}
\begin{item}
a point $p \in \Omega$ is a critical point of $f$ iff $p$ is a critical point of $1/f$, 
\end{item}
\begin{item}
a point $p \in \Omega$ is a degenerate critical point of $f$ iff $p$ is a degenerate critical point of $1/f$. 
\end{item}
\end{enumerate} 
Thus, it suffices to show that
for almost all 
$(x_1 ,\dots,x_{i-1},
x_{i+1} ,
\dots,x_{n})$
a function $\frac{1}{ \SINR(x_i, \cdot) }$ has no degenerate critical points.

We take 
$$W = \{(p,(x_1,\dots,x_{i-1},x_{i+1},\dots,x_d)) \vert $$
$$\vert p,x_1,\dots,x_{i-1},x_{i+1},\dots,x_d \text{ are distinct points in } \R^d \setminus \{x_i\},$$
$$\text{ and } x_1,\dots,x_{i-1},x_{i+1},\dots,x_d \text{ do not lie on all on the same line}  \}$$
and we define $$U:W \rightarrow \R$$
by $$U(p,(x_1,\dots,x_{i-1},x_{i+1},\dots,x_{n})) = 
\frac{1}{\SINR(x_i, p)}.$$
For $h \in \{1,\dots,i-1,i+1,\dots,n\}$ consider a matrix $d \times d$ matrix $M_{h}$ 
given by  
$$M_h = 
\begin{bmatrix}
    U''_{p_1x_{h1}} & U''_{p_1x_{h2}} & \dots & U''_{p_1x_{hd}} \\
    U''_{p_2x_{h1}} & U''_{p_2x_{h2}} & \dots & U''_{p_2x_{hd}} \\
    \hdotsfor{4} \\
     U''_{p_dx_{h1}} & U''_{p_dx_{h2}} & \dots & U''_{p_dx_{hd}}
\end{bmatrix}
$$
By Proposition \ref{morseGen} it suffices to show that the matrix 
$$[M_1,\dots,M_{i-1},M_{i+1},\dots,M_{n}]$$
has rank $d$ everywhere in $W$. 
Suppose that's not true. Then for some $$(p,(x_1,\dots,x_{i-1},x_{i+1},\dots,x_{n})) \in W$$
there exists $\ww \in \R^{d} \setminus \{0\}$ such that 
\begin{equation}w^TM_{h} = 0,\label{wMeq0}\end{equation}
for every $h \in {1,\dots,i-1,i+1,\dots,n}$.

Next, we need to explain what happens to a row vector when it gets multiplied by $M_{h}$ from the right.
For every $h \in    \{1,\dots,i-1,i+1,\dots,n\}$ and $k \in \{1,\dots,d\}$
$$U'_{x_{hk}}(p) =  \aa\frac{ \psi_h \vert  \vert x_h - p \vert \vert^{-(\alpha + 2)} }
{\psi_i \vert  \vert x_i - p \vert \vert^{-\alpha} }(p_k - x_{hk}).$$
We denote 
$$Q_h := \aa\frac{ \psi_h \vert  \vert x_h - p \vert \vert^{-(\alpha + 2)} }
{\psi_i \vert  \vert x_i - p \vert \vert^{-\alpha} },$$
so we have 
$$U'_{x_{hk}}(p) =   (p_k - x_{hk})Q_h.$$
For every $h \in \{1,\dots,i-1,i+1,\dots,n\}$ and $k,j \in \{1,\dots,d\}$ we have 
$$U''_{x_{hk}p_j}(p) = P^{h,k,j}_1 + P^{h,k,j}_2 + P^{h,k,j}_3,$$
where 
$$P^{h,k,j}_1 = \begin{cases} 0 & \text{if } j \neq k \\ Q_h  & \text{if } j = k \end{cases},$$ 
$$P^{h,k,j}_2 =  -(\aa + 2)(p_k - x_{hk})(p_j - x_{hj})\vert \vert x_h - p \vert \vert^{-2}Q_h,$$
$$P^{h,k,j}_3 =   \aa(p_k - x_{hk})(p_j - x_{ij})\vert \vert x_i - p \vert \vert^{-2} Q_h.$$

Thus $M_h$ can be presented as 
$$M_h = M_{h1} + M_{h2} + M_{h3},$$
such that for every $v = (v_1,\dots,v_d)^T \in \R^{d}$ we have 
$$v^T M_{h1} = v^TQ_h,$$
$$v^T M_{h2} = -(\aa + 2)\vert \vert x_h - p \vert \vert^{-2}v^T(p-x_h)(p-x_h)^TQ_h,$$
$$v^T M_{h3} =   \aa  \vert \vert x_i - p \vert \vert^{-2}v^T(p-x_i)(p-x_h)^TQ_h.$$
Note that for every $v \in \R^{d}$ we have that 
$$v^T M_{h2} = \ll_2(v) (p-x_h)^T,$$
$$v^T M_{h3} = \ll_3(v) (p-x_h)^T,$$
for some $\ll_2(v), \ll_3(v) \in \R$. Thus from $\ww^TM_{h} = 0$ we have that 
$$\ww^T M_{h1} = -(\ll_2(w) + \ll_3(w)) (p-x_h)^T,$$
and we can conclude that for every $h \in {1,\dots,i-1,i+1,\dots,n}$ 
there exists $c(h) \in \R \setminus \{0\}$ such that 
$$\ww = c(h) (p-x_h).$$
This implies that all $x_1,\dots,x_{i-1},x_{i+1},\dots,x_d$ all lie on the same line. 
But that's the opposite of what is stated in the definition of $W$.
\end{proof} 

\section{Newtonian Point Masses with a Central Force}
This model is similar to Maxwell's problem for the case $m = 1$ 
with all the charges having the same sign except there is an 
additional quadratic term. More precisely,   
let $d \in \{1,2,3,\dots\}$.  
Suppose that we have points $x_1,\dots,x_n \in \R^d$.
And numbers $m_1,\dots,m_n  > 0$ 
which symbolize point masses located in those points. 
%Those charges create an electrostatic field whose 
%potential in a given point $p \in \R^{d}$ is given by 
Consider a function 
$F = F^{(x_1,m_1),\dots,(x_n,m_n)}:\R^{d} \rightarrow \R$ given by
\begin{equation}F(p) = \frac{1}{2}||p||^2 + \sum_{i = 1}^{n}\frac{m_i}{\vert \vert p - x_i \vert \vert}.
\label{NewtonPotential}\end{equation}
  Once again we are interested in  an upper bound on the number 
  of critical points of $F$. 

This model arises in the study of the restricted $(n + 1)$-body
problem. We refer the reader to \cite{arustamyan2021number} 
for details and the relevant bibliography. 
 
\subsection{Rough summary of our results for Newtonian Point Masses with a Central Force} \label{NewtonSummary}
We give an exponential upper bound for the number of isolated 
critical points for any dimension $d$. More precisely the number 
of isolated critical points of $F$ does not exceed 
$$(1 + 3n)(1 + 6n)^{d + n}.$$

\subsection{Previous results}
%\subsubsection{Critical points of a potential generated by Newtonian point masses}
Arustamyan et al. \cite{arustamyan2021number} claim an exponential upper bound 
for the number of isolated 
critical points for the planar case of this problem. 
  
In astrophysics, there is a related open problem 
of finding upper bounds on the 
number of possible images in gravitational microlensing, 
see \cite{khavinson2008fromthefundamental, petters2010gravity}.
S. Perry \cite{perry2021upper} gives the only
 known upper bound for the general case. 
 
\subsection{Our results for Newtonian Point Masses with a Central Force}
The following theorem formalizes the informal summary given in Subsection \ref{NewtonSummary}. 

\begin{theorem}\label{Newton}
Let $d \in \{1,2,3,\dots\}$.
Let $x_1 = (x_{11},\dots,x_{1d}),\dots,x_{n} = (x_{n1},\dots,x_{nd})$ 
be points in $\R^d$,
$p  = (p_1,\dots,p_d) \in \R^d \setminus \{x_1,\dots, x_n\}$,
$m_1,\dots,m_n \in \R \setminus \{0\}$
and let $F = F^{(x_1,m_1),\dots,(x_n,m_n)}$ be the function given by (\ref{NewtonPotential}). Then,
\begin{enumerate}
\begin{item}
Point $p$ is a critical point of $F$ iff there exist $\ss_1,\dots,\ss_n > 0$
satisfying
%such that $p_1,\dots,p_d,\ss_1,\dots,\ss_n$ satisfy 
$$
\ss_j^2  \sum_{1\le k \le d}(p_k - x_{jk})^2 = 1 
\text{, for every $1\le j \le n$, and}
$$
\begin{equation}
p
% \prod_{ 1 \le j \le n }
%\ss_j^{3}
 -
\sum_{i=1}^{n}
\Big(
%q_i(p_m - x_{im})
m_i(p - x_{i})
%\prod_{\substack {  1 \le j \le n \\ j \neq i} }
% \ss_j^{3}
\ss_i^{3}
\Big) = 0.
%\text{, for every $1\le m \le d$.}
\label{Newton-EEE}\end{equation}\label{Newton-p1}
\end{item}
\begin{item}
$F$ has at most 
$$4 \cdot 7^{d + n}$$
 isolated critical points.
\label{Newton-p2}
\end{item}
\end{enumerate}
\end{theorem}

%%%%%%%%%%%%%%%%%%%%%%%%%%%%%%%%%%
%%%%%%%%%%%%%%%%%%%%%%%%%%%%%%%%%
%%%%%%%%%%%%%%%%%%%%%%%%%%%%%%kkkk iiii

\begin{proof}[Proof of Theorem \ref{Newton}(\ref{Newton-p1})]
We remind that  
$$F(p) = \frac{1}{2}||p||^2 + \sum_{i = 1}^{n}\frac{m_i}{\vert \vert p - x_i \vert \vert}.$$
We differentiate $V$:
$$F'_{p_k}  =  p_k - \sum_{i = 1}^{n}
\frac{m_i(p_k - x_{ik})}
{\vert \vert p - x_i \vert \vert^{3}},$$ 
$$\nabla F =  p - \sum_{i = 1}^{n}
\frac{m_i(p - x_{i})}
{\vert \vert p - x_i \vert \vert^{3}},$$
Thus, $\nabla F(p) = 0$ is equivalent to  
%$$
%p \prod_{ 1 \le j \le n }
%\vert \vert p - x_j \vert \vert^{3}
%-\sum_{i=1}^{n}
%\Big(
%%q_i(p_m - x_{im})
%m_i(p - x_{i})
%\prod_{\substack {  1 \le j \le n \\ j \neq i} }
%\vert \vert p - x_j \vert \vert^{3}
%\Big) = 0.$$
%Which is equivalent to 
(\ref{Newton-EEE}). 
\end{proof}

\begin{proof}[Proof of Theorem \ref{Newton}(\ref{Newton-p2})]
Denote 
 $$N = 4 \cdot 7^{d + n}.$$
We will argue by contradiction. 
Suppose that $F$ has at least $N + 1$ isolated critical points.
Then by Theorem \ref{Newton}(\ref{Newton-p1})
the zero set of the system (\ref{Newton-EEE}) (considered as a polynomial 
system in $d + n$ variables $p_1,\dots,p_k,\ss_1,\dots,\ss_n$)
has at least $N + 1$ connected components. 

The degree of every polynomial in this system is $\le 4$.
Thus,
by Thom-Milnor theorem (see Corollary \ref{ThomMilnerCor}) we 
have that the total number of the connected components of the zero set 
does not exceed 
$$4 \cdot 7^{d + n} = N,$$
so we have a contradiction.   
\end{proof}

\subsection{Existence of non-isolated critical points} \label{NewtonNonIso}
There is only one known example of a configuration with non-isolated
critical points: we take $n = 1$ and put the only one point mass 
we have into the origin.    

It is also known that for almost all configurations 
there is no degenerate critical points, see \cite[Theorem 2.1]{arustamyan2021number}
(authors of \cite{arustamyan2021number} state their result only for $d = 2$
but their proof works for every $d$).

\section{Central configurations} \label{SecCentral}
Let $d,n \in \{1,2,3,\dots\}$, $m_1, \dots, m_n > 0$ and 
let $x_1,\dots,x_n$ be distinct points in $\R^d$.
We say that a system of point masses $(x_1, m_1),\dots,(x_n, m_n)$ 
is a central preconfiguration if there exists $\ll > 0$   
such that for every $i = 1,\dots,n$,
$$\ll x_i = \sum_{\substack {  1 \le j \le n \\ j \neq i} }(m_i\vert\vert x_i - x_j\vert \vert^{-3}(x_i - x_j)).$$
We say that a central preconfiguration is normalized if $\ll$ in the above system 
is equal to $1$. (Every central preconfiguration can be normalized by scaling.) 

We say that normalized central preconfigurations 
$(x_1, m_1),\dots,(x_n, m_n)$ and $(x'_1, m_1),\dots,(x'_n, m_n)$ 
are equivalent if there
exists an orientation and origin-preserving isometry
$$F:\R^d \rightarrow \R^d$$ such that 
$F(x_i) = x'_i$ for every $i = 1,\dots,n$. 
%and such that there exists $C > 0$ such that for every $u,v \in \R^d$
%$$\vert\vert F(u) - F(v)\vert\vert = C\vert\vert u - v \vert\vert.$$
Central configuration is an equivalence class of normalized central preconfigurations.

%(This implies that $F$ is linear.) Then, $(x_1, m_1),\dots,(x_n, m_n)$ is 
%a central configuration iff $(F(x_1), m_1),\dots,(F(x_n), m_n)$ is a central configuration.
%We will call two such configurations equivalent. 

We say that a c 
central configurations $K$ is isolated 
iff there exists $\ee > 0$ such that for every normalized central preconfiguration 
$(x_1, m_1),\dots,(x_n, m_n)$ from $K$ and every normalized central preconfiguration 
$(y_1, m_1),\dots,(y_n, m_n)$ which does not belong to $K$ we have 
$$\max_{1 \le i \le n} \vert \vert x_i - y_i \vert \vert \ge \ee.$$
 
%\subsection{Rough summary of our results for central configurations} \label{CentralSummary}
%We give an exponential upper bound for the number of isolated 
%equivalence classes of normalized
%central configurations. Namely their number does not exceed 
%$$....$$

%\subsection{Previous results and related results}

\subsection{Our results for central configurations}
The following theorem is very similar to the result of R. Kuzmuna \cite[Theorem 3]{kuzmina1977upper}.
R. Moekel's estimates \cite{moeckel1985relative, moeckel2001generic} for the number of central 
configurations also utilize Thom-Milnor theorem.   

\begin{theorem}\label{central}
Let $d,n \in \{1,2,3,\dots\}$, $m_1, \dots, m_n > 0$
There are at most 
$$4 \cdot 7^{\frac{n(n-1)}{2} + nd - 1}$$
isolated central configurations with 
those parameters.
\end{theorem}

\begin{proof}
Let $x_1 = (x_{11},\dots,x_{1d}),\dots,x_n = (x_{n1},\dots,x_{nd}) \in \R^d$.
A system of point masses $(x_1, m_1),\dots,(x_n, m_n)$ is 
a normalized central preconfiguration iff  there exist
$\{\ss_{ij}\}_{1 \le i \neq j \le n}$
such that $\ss_{ij} = \ss_{ji} > 0$ and 
satisfying 
%such that $p_1,\dots,p_d,\ss_1,\dots,\ss_n$ satisfy 
$$
\ss_{ij}^2 \sum_{1\le k \le d}(x_{ik} - x_{jk})^2 = 1 
\text{, for every $1\le i < j \le n$ and}
$$
\begin{equation}x_i = \sum_{\substack {  1 \le j \le n \\ j \neq i} }(m_i\ss_{ij}^{3}(x_i - x_j)).
\label{Central-EEE}\end{equation}

Denote 
 $$N = 4 \cdot 7^{\frac{n(n-1)}{2} + nd - 1}.$$
We will argue by contradiction. 
Suppose that there are at least $N + 1$ isolated  central configurations.
Then  
the zero set of the system (\ref{Central-EEE}) (considered as a polynomial 
system in $\frac{n(n-1)}{2} + nd$ variables 
$\{\ss_{ij}\}_{1 \le i < j \le n}$, 
$\{x_{ij}\}_{1 \le i \le n, 1 \le j \le d}$
)
has at least $N + 1$ connected components. 
The degree of every polynomial in this system is $\le 4$.
Thus,
by Thom-Milnor theorem (see Corollary \ref{ThomMilnerCor}) we 
have that the total number of the connected components of the zero set 
does not exceed 
$$4 \cdot 7^{\frac{n(n-1)}{2} + nd - 1} = N,$$
so we have a contradiction.   
\end{proof}
 
%According to Wikipedia \cite{enwiki:1133031763}: "a central configuration is a system of point masses with the property that each mass 
%is pulled by the combined gravitational force of the system directly towards the center of mass.
%...
%Two central configurations are considered to be equivalent if they are similar, that is, they can 
%be transformed into each other by some combination of rotation, translation, and scaling". 
%
%With this in hand we can ...

\subsection{Existence of non-isolated central configurations}
It is unknown if non-isolated central configurations exist. This is a high-profile 
open problem known as 6th Smale's problem, see \cite{smale1998mathematical}.

M. Hampton and R. Moeckel \cite{hampton2006finiteness} showed that 
all central configurations are isolated for $n = 4$.
A. Albouy and V. Kaloshin \cite{albouy2012finiteness} proved the same 
for $d = 2$ and $n = 5$ for almost all sets of masses. 

%To be honest the author of the present text is not really familiar with 
%this field. 
%Thus, we will skip its overview and instead 
%
For more information on the field we
refer the interested reader to the introduction in \cite{albouy2012finiteness}.

\subsection*{Acknowledgements}
The author is thankful to Andrei Alpeev and Pasha Galashin for discussing 
 early versions of the paper and pointing out a serious bug. I am grateful to
 Mathoverflow users Wille Liou, Gro-Tsen, and user43326 for helping me with 
 my questions on algebraic geometry. % \cite{442829, 449165}.
 The author is thankful to Alain Albouy for pointing me to the works of Kuzmina and Moeckel.
\bibliography{circle}

\bibliographystyle{plain}

\end{document}